\documentclass[final,3p,times]{elsarticle}

\usepackage{graphicx,amssymb,amsmath,amsthm,bm,accents,color,subfig,hyperref}

\newtheorem{thm}{Theorem}
\newtheorem{lemm}[thm]{Lemma}

\journal{Nonlinear Analysis: Real World Applications}

\begin{document}

\begin{frontmatter}

\title{Comments on: ``Energetic balance for the Rayleigh--Stokes problem of an Oldroyd-B fluid'' [Nonlinear Anal. RWA 12 (2011) 1]}

\author{Ivan C. Christov}
\ead{christov@alum.mit.edu}
\ead[url]{http://alum.mit.edu/www/christov}
\address{Department of Engineering Sciences and Applied Mathematics, Northwestern University, Evanston, IL 60208-3125, USA}

\begin{abstract}
We point out that an erroneous derivation in the recent paper [Fetecau et al., Nonlinear Anal. RWA 12 (2011) 1] yields a correct solution by accident. Additionally, a number of misrepresentations and inaccuracies in the latter recent paper are identified, corrected and/or clarified in this Comment. Finally, a listing of recent papers in this journal that make a mistake applying the Fourier sine transform, and thus present erroneous solutions, is given as an Appendix.
\end{abstract}


\end{frontmatter}


\section{Introduction}\label{sec:intro}
The recent paper \cite{Fetecau} considers the following partial differential equation:
\begin{equation}
\frac{\partial u}{\partial t} + \lambda \frac{\partial^2 u}{\partial t^2} = \nu\frac{\partial^2 u}{\partial y^2} + \nu\lambda_r\frac{\partial^3 u}{\partial y\partial t\partial y},\qquad y,t>0,
\label{eq:pde}
\end{equation}
which arises in the mathematical description of unidirectional planar motion of an Oldroyd-B (non-Newtonian) fluid resting upon an impermeable infinite plate situated at $y=0$ (see, e.g., \cite{Christov} and the references therein for a complete derivation), with $u=u(y,t)$ being the fluid's velocity component tangential to the plate. Here, $\nu$, $\lambda$ and $\lambda_r$ are constants. Furthermore, Eq.~\eqref{eq:pde} is supplemented with the following initial and boundary conditions \cite[Eqs.~(8)--(9)]{Fetecau}:
\begin{equation}
u(y,0) = \frac{\partial u}{\partial t}(y,0) = 0,\qquad y>0;\qquad u(0,t) = U(t),\quad u \to 0 \;\;\text{as}\;\; y\to\infty,\qquad t>0.
\label{eq:ibc}
\end{equation}
Here, we have omitted the condition $\partial u/\partial y\to0$ as $y\to\infty$ from \cite[Eq.~(9)]{Fetecau} because it is extraneous \cite{Christov}.

The plate's velocity $U(t)$ \emph{might appear} to be an arbitrary function at this point, but \emph{it is not} \cite{ChristovChristov}. This is because, in following \cite[Eq.~(8)]{Fetecau}, we have made a poor choice of notation in expressing the initial-boundary conditions. In fact, this has led to a significant number of errors when solving Eqs.~\eqref{eq:pde}--\eqref{eq:ibc} by integral transform methods \cite{Christov,ChristovChristov,ChristovJordan}. Specifically, for Stokes' first problem, which is called the ``Rayleigh--Stokes problem'' in \cite{Fetecau}, the initial and boundary conditions are always ``incompatible'' at $(y,t) = (0,0)$ because of the requirement that the plate's motion be \emph{sudden} (see the discussion in \cite{Christov}). In other words, to ensure that the plate's velocity is identically zero prior to start-up but arbitrary (not identically zero) following start-up, as stipulated in Eq.~\eqref{eq:ibc}, one must write $U(t) := \widetilde U(t) H(t)$, where $\widetilde U(t)$ is a smooth function that we are free to specify, and $H(t)$ is the Heaviside unit step function.\footnote{It is noteworthy that, in Tanner's 1962 paper \cite{Tanner} on Stokes' first problem for the Oldroyd-B fluid, the initial-boundary conditions are correctly formulated in the manner discussed here (see \cite[Eq.~(17)]{Tanner}). This correct classical work is not credited in \cite{Fetecau}.} This is \emph{not} a ``new'' version of boundary condition, it is simply {the} mathematically-precise statement of the very same initial-boundary conditions for a suddenly moved plate stipulated in Eq.~\eqref{eq:ibc}. Thus, for example, for Stokes' first problem, one takes $\widetilde U(t) = const.$, while for Stokes' second problem one may take $\widetilde U(t) = \cos t$. However, since $H(t)$ is {discontinuous} at $t=0$, i.e., $\lim_{t\to 0^+}H(t) \ne \lim_{t\to0^-} H(t)$, it follows that  $U(t)$ is {discontinuous at $t=0$ as well}. Thus, great care must be taken in properly evaluating derivatives and integrals of the function $U(t)$ as written in Eq.~\eqref{eq:ibc}.

For what follows, it is important to realize that $\widetilde U(t)$ is a well-behaved function in physically-relevant problems. For example, we can take $\widetilde U(t)$ to be continuously differentiable for all $t\in\mathbb{R}$ without loss of generality. Then, $\widetilde U(0)$ is a well-defined quantity \emph{unlike} $U(0)$. Therefore, the claim that ``for a greater generality we consider the boundary condition $u(0,t) = U(t)$ with $U(0) = 0$'' \cite[p.~3]{Fetecau} is without merit when $\widetilde U(t) = const.$ because $U(t)$ is discontinuous at $t=0$. It appears that left-continuity of $u(\cdot,t)$ at the boundary is being implicitly enforced, without justification, in \cite{Fetecau}. To the contrary, in the theory of non-smooth ordinary differential equations (notice the lack of smoothness on the right-hand side of Eq.~\eqref{eq:new_ode} below), it is necessary to enforce \emph{right}-continuity at $t=0$ (see, e.g., \cite[p.~102]{AB08} or \cite[Chap.~15]{D74}), whence $U(0) \ne 0$. 

\section{How to correctly find the solution in the Fourier--Laplace transform domain}
Following \cite{Fetecau}, the next step in the analysis is applying the Fourier sine transform (in $y$ with $\xi$ as the ``dummy'' variable) to Eq.~\eqref{eq:pde}. Thus, it is claimed that the following ordinary differential equation (ODE) is obtained:
\begin{equation}
\lambda \frac{\partial^2 u_\mathrm{s}}{\partial t^2} + (1 + \alpha\xi^2) \frac{\partial u_\mathrm{s}}{\partial t} + \nu\xi^2 u_\mathrm{s} = \xi \sqrt{\frac{2}{\pi}}\left[\nu U(t) + \alpha U'(t) \right], \qquad t>0,
\label{eq:ode}
\end{equation}
where $u_\mathrm{s}(\xi,t)$ denotes the image of $u(y,t)$ in the Fourier domain and $\alpha := \nu\lambda_r$. Equation~\eqref{eq:ode} is subject to the transformed initial conditions $u_\mathrm{s}(\xi,0) = \frac{\partial u_\mathrm{s}}{\partial t}(\xi,0) = 0$. 

In order to be able to solve such an ODE, one {must} know \emph{a priori} the degree of smoothness of the right-hand side. Does the classical solution of the ODE exist or should a solution be sought in the sense of distributions \cite[Chap.~2]{Kanwal}? This is not made clear in Eq.~\eqref{eq:ode} above or in \cite[Eq.~(10)]{Fetecau}. Instead, a cavalier application of the Laplace transform $\mathcal{L}$ (in $t$ with $q$ as the ``dummy'' variable and an over-line denoting the image of a function in the Laplace domain) is made in \cite{Fetecau}. Additionally, the authors implicitly use the relation $\mathcal{L}\{U'(t)\}(q) = q\overline{U}(q) - U(0)$ (incorrectly assuming, once again, that $U(0)=0$), which is only true for \emph{continuously differentiable functions} \cite[\S8.1, Property 3]{Kanwal}. Unfortunately, $U(t)\equiv\widetilde U(t)H(t)$ fails to be continuous at $t=0$, in general, as we explained above. Hence, the latter Laplace transform identity does not apply; moreover, the point value $U(0)$ to be used in it is  ill-defined for Stokes' first problem and related start-up problems. Consequently, the unjustified and erroneous assumption that $U(0)=0$ renders the Fourier--Laplace domain solution of Eq.~\eqref{eq:ode}, i.e., \cite[Eq.~(12)]{Fetecau}, incorrect.

To remedy the situation, let us first note that the term $U'(t)$, which comes form the mixed third-order derivative in Eq.~\eqref{eq:pde}, constitutes a  distributional derivative. Then, we have the following result:
\begin{lemm}
$U'(t) = \widetilde U'(t)H(t) + \widetilde U(0)\delta(t)$, where $\delta(t)$ is the Dirac delta distribution (generalized function).
\label{lemm:derivative}
\end{lemm}
\begin{proof}
By the general result \cite[\S8.3, Eq.~(4)]{Kanwal} and the proper definition of the quantity $U(t)$ from Sec.~\ref{sec:intro}, $U'(t) \equiv \left[\widetilde U(t)H(t)\right]' = \widetilde U'(t)H(t) + [\![ U ]\!] \delta(t)$, where $[\![ U ]\!] := \lim_{t\to0^+} U(t) - \lim_{t\to0^-} U(t)$. From the fact that $\widetilde U(t)$ is continuous at $t=0$, we have $[\![ U ]\!] = \widetilde U(0)\cdot 1 - \widetilde U(0) \cdot 0 = \widetilde U(0)$.
\end{proof}

By Lemma~\ref{lemm:derivative}, upon applying the Fourier sine transform to Eq.~\eqref{eq:pde} subject to Stokes-type initial-boundary data, we obtain
\begin{equation}
\lambda \frac{\partial^2 u_\mathrm{s}}{\partial t^2} + (1 + \alpha\xi^2) \frac{\partial u_\mathrm{s}}{\partial t} + \nu\xi^2 u_\mathrm{s} \stackrel{\mathrm{dist}}{=} \xi \sqrt{\frac{2}{\pi}}\left\{\nu \widetilde U(t)H(t) + \alpha \left[\widetilde U'(t)H(t) + \widetilde U(0) \delta(t)\right] \right\},
\label{eq:new_ode}
\end{equation}
rather than what is found in Eq.~\eqref{eq:ode}. Here, we have deliberately introduced the symbol $\stackrel{\mathrm{dist}}{=} $ to stress that the equality is in the sense of distributions \cite[\S2.3]{Kanwal}. Note that $\widetilde U(0) \delta(t) \stackrel{\mathrm{dist}}{=} \widetilde U(t) \delta(t)$ also \cite[\S2.5, Example 2]{Kanwal}; so, equivalently, we may write $U'(t) = \widetilde U'(t)H(t) + \widetilde U(t)\delta(t)$, which is the result of formally differentiating $U(t)\equiv\widetilde U(t)H(t)$. Now, to apply the Laplace transform to the right-hand side of the Eq.~\eqref{eq:new_ode}, we first need to derive two identities.

\begin{lemm}
$\mathcal{L}\left\{\widetilde U'(t)H(t)\right\}(q) = q\overline{\widetilde U}(q) - \widetilde U(0).$
\label{lemm:UpH}
\end{lemm}
\begin{proof}
Following \cite[\S8.2]{Kanwal},
$\int_0^\infty \widetilde U'(t) H(t) \mathrm{e}^{-qt} \,\mathrm{d}t = \int_0^\infty \widetilde U'(t) \mathrm{e}^{-qt} \,\mathrm{d}t  \equiv \mathcal{L}\left\{\widetilde U'(t)\right\}(q).$
Then, from \cite[\S8.1, Property 3]{Kanwal} and the fact that $\widetilde U(t)$ is a continuously differentiable function for all $t\in\mathbb{R}$ (e.g., a polynomial or trigonometric function in Stokes-type problems), we have the desired result.
\end{proof}

\begin{lemm}
$\mathcal{L}\left\{U'(t)\right\}(q) \equiv \mathcal{L}\left\{\widetilde U'(t)H(t) + \widetilde U(0)\delta(t)\right\}(q) = q\overline{\widetilde U}(q).$
\end{lemm}
\begin{proof}
From \cite[\S8.1, Property 1]{Kanwal}, we have $\mathcal{L}\left\{\widetilde U'(t)H(t) + \widetilde U(0)\delta(t)\right\}(q) = \mathcal{L}\left\{\widetilde U'(t)H(t)\right\}(q) + \widetilde U(0)\mathcal{L}\left\{\delta(t)\right\}(q)$, which is just $q\overline{\widetilde U}(q) - \widetilde U(0) + \widetilde U(0) \cdot 1 = q\overline{\widetilde U}(q)$ thanks to Lemma~\ref{lemm:UpH} and \cite[\S8.2, Example 2(d)]{Kanwal}.
\end{proof}

Note that this result holds for any smooth $\widetilde U(t)$ and does \emph{not} require the unjustifiable assumption that $U(0) = 0$ employed in \cite{Fetecau} to ``make things work out.''  In fact, we have shown the much stronger result that the value of $\widetilde U(t)$ at $t=0$ is irrelevant as long as $\widetilde U(0)$ exists and is well-defined.

Finally, upon acknowledging the initial conditions, we arrive at the correct solution of Eq.~\eqref{eq:pde} in the Fourier--Laplace transform domain (with $\overline{u_\mathrm{s}}(\xi,q)$ denoting the image of $u(y,t)$ there):
\begin{equation}
\overline{u_\mathrm{s}}(\xi,q) = \xi\sqrt{\frac{2}{\pi}}\frac{\nu + \alpha q}{\lambda q^2 + (1 + \alpha \xi^2) q + \nu\xi^2}\overline{\widetilde U}(q).
\label{eq:fltd}
\end{equation}
Though Eq.~\eqref{eq:fltd} is deceivingly similar to \cite[Eq.~(12)]{Fetecau}, the former is the true Fourier--Laplace transform domain solution, while the latter is incorrect in general. Given the contradictory assumptions and lack of proof, it appears that \cite[Eq.~(12)]{Fetecau} is similar to Eq.~\eqref{eq:fltd} only because the special case given in \cite[Eq.~(3)]{ChristovJordan} was ``reverse engineered.''

\section{Further deficiencies}
In \cite[Sec.~2]{Fetecau}, it is stated that ``in order to solve a well-posed problem for such fluids one has to require an additional initial condition apart from the requirement that the fluid is initially at rest'' \cite[p.~2]{Fetecau}. This statement concerns the initial condition $\frac{\partial u}{\partial t}(y,0) = 0$ in Eq.~\eqref{eq:ibc}. However, for Stokes' first problem of the impulsively moved infinite plate, the fluid being at rest initially, i.e., $u(y,t\le 0^-) \equiv 0$ (identically), trivially implies it has zero initial velocity $u(y,0)=0$, zero initial acceleration $\frac{\partial u}{\partial t}(y,0) = 0$, zero initial jerk $\frac{\partial^2 u}{\partial t^2}(y,0) = 0$, zero initial jounce $\frac{\partial^3 u}{\partial t^3}(y,0) = 0$, and so on. Of course, these time derivatives are one-sided (with $y>0$) because the solution itself possesses no time derivative across the plane $t=0$ as shown by the discussion above (see also \cite[p.~577]{Tanner} and \cite[p.~719]{Christov}). None of these initial conditions are ``additional assumptions'' beyond the assumption of the fluid being initially at rest, which was made in the problem's formulation. Just because Eq.~\eqref{eq:pde} is second-order in time and requires the use of two initial conditions stemming from the assumption that the fluid is initially at rest (rather than the single initial condition needed for, e.g., the second grade or viscous Newtonian fluids) does \emph{not} mean that ``additional assumptions'' were made.

Furthermore, an unreduced form of the $yt$-domain solution is presented in \cite[Eq.~(17)]{Fetecau}, which appears to be able to give a complex-valued velocity if $(1+\alpha\xi^2)^2 < 4\nu\lambda\xi^2$ in the expressions for $r_{1,2}$. The concerned reader is referred to \cite[Eq.~(7)]{ChristovJordan}, where it is shown that there are three cases to be considered and the solution is real-valued for all three.

In \cite[Sec.~4.1]{Fetecau}, it is claimed by the authors that they ``consider the case when the dimensionless relaxation and retardation times, $\lambda/t$ and $\lambda_r/t = \alpha / (\nu t)$ are much less than one'' \cite[p.~6]{Fetecau}. However, this is a false statement. What is provided are results \emph{for times that are long compared to both the relaxation and retardation time scales set by $\lambda$ and $\lambda_r$, respectively.} The dimensionless relaxation and retardation times, depending on the choice on non-dimensionalization scheme, could be given by, e.g., $\lambda^* = \lambda U_0^2/\nu$ and $\lambda_r^* = \lambda_r U_0^2/\nu$, where $U_0$ is a characteristic velocity. These are never introduced in \cite[Sec.~4.1]{Fetecau}.

In \cite[Sec.~4.2]{Fetecau}, a scaling argument is given about the case of ``$\lambda/t,\; \lambda_r/t \gg 1$,'' leading to the claim that, in this limit, Eq.~\eqref{eq:pde} reduces to the analogous equation for a viscous Newtonian fluid if additionally $\lambda = \lambda_r$. However, it is clear from \cite[Eq.~(7)]{ChristovJordan} that, when $\lambda = \lambda_r$, the solution to Stokes' first problem for the Oldroyd-B fluid is \emph{identical} to the solution of Stokes' first problem for the viscous Newtonian fluid, independently of the sizes of $\lambda/t$ and $\lambda_r/t$. The result in \cite{ChristovJordan} is mathematically stronger  because it is a statement about the solution for $\lambda=\lambda_r$, regardless of any asymptotic approximation to Eq.~\eqref{eq:pde}. The result is obvious once the solution to Eq.~\eqref{eq:pde} is properly reduced, unlike the expression presented in \cite[Eq.~(17)]{Fetecau}.

In addition, it is claimed that the correct solution for a second grade fluid executing the same motion (formally, the $\lambda\to0$ limit of the Oldroyd-B solution) ``can be also obtained from [18, Eq.~(2.5)] (integrating by parts the last term and taking $V(t) = UH(t)$)'' \cite[p.~4]{Fetecau}. Of course this is untrue as \cite[Ref.~18]{Fetecau} is known to contain erroneous derivations \emph{and} solutions \cite{ChristovChristov}. Therefore, claiming \emph{ex post facto} that a correct solution could be extracted from a wrong solution (once it is known what the correct solution is supposed to be) is fallacious reasoning. At any rate, the derivation of \cite[Ref.~18, Eq.~(2.5)]{Fetecau} suffers from the same logical inconsistencies explained in detail above. The reader is referred to \cite{Christov,ChristovChristov} for more exposition on the issue.

Finally, the list of references provided in \cite{Fetecau} fails to properly attribute previous results to the appropriate works in the literature. For the limiting case of a Maxwell fluid executing the same motion (formally, the $\lambda_r\to0$ limit in the Oldroyd-B solution), the energy analysis of Stokes' first problem is attributed to a 2007 paper \cite[Ref.~6]{Fetecau} co-authored by two of the authors of \cite{Fetecau}. The latter paper \cite[Ref.~6]{Fetecau}, the one under current discussion \cite{Fetecau} and a third paper by two of the authors of \cite{Fetecau} in {\it Int.\ J.\ Non-linear Mech.}\ (vol.\ 44, pp.\ 862--864) in 2009, all neglect to mention that the energy analysis of Stokes' first problem for Maxwell fluids had already been completed in full in 2005 \cite{JordanPuri}.

\section{Closure}
The failure to pose the initial-boundary conditions in Eq.~\eqref{eq:ibc} in a mathematically-precise form, combined with the inherently contradictory assumptions on the form and smoothness of the poorly-defined function $U(t)$ from Eq.~\eqref{eq:ibc}, has lead to a large number of wrong solutions being published (see, e.g., \cite{ChristovChristov,ChristovJordan,Jordan} for some proper corrections and also \cite{Christov} for a longer listing of erroneous papers). Even a recent Erratum \cite{FetecauErratum} claims that ``for $V(t) = V\cos(\omega t)$ ... the solutions corresponding to the motion induced by an oscillating ... plate are recovered'' \cite[p.~360]{FetecauErratum} ($V$ therein being $U$ from the present Comment). Of course, $\lim_{t\to0^+} \cos(\omega t) = 1 \ne 0$, so the solution corresponding to the motion induced by a suddenly oscillated plate is \emph{not} recovered. This is because of the lack of $H(t)$ multiplying this expression, meaning its derivate calculated as in \cite[Eq.~(5)]{FetecauErratum} is generally incorrect (just as Eq.~\eqref{eq:ode} above is generally incorrect).

Finally, it is also important to note that Eq.~\eqref{eq:pde} is \emph{linear} and has been solved (correctly) subject to the initial-boundary conditions in Eq.~\eqref{eq:ibc} as early as 1956 \cite{Morrison} (see also \cite[Sect.~2.2]{Christov}). Therefore, it is unclear how the recent paper \cite{Fetecau} under discussion here is ``demonstrating the relevance and applicability of \emph{nonlinear} techniques'' (emphasis added) \cite{NONRWA}.

\section*{Acknowledgments}
The author would like to express his gratitude to Prof.\ C.\ I.\ Christov and Dr.\ P.\ M.\ Jordan for enlightening discussions and encouragement.


\appendix
\section*{Appendix. Other erroneous recent papers on simple flows of non-Newtonian fluids in Nonlinear Anal.\ RWA}
Papers published in this journal that make the mistake of incorrectly differentiating the suddenly-moved plate's velocity (and therefore provide incorrect, unphysical solutions as discussed in \cite{Christov,ChristovChristov,ChristovJordan,Jordan}) include, but are not limited to, the following:
\begin{enumerate}
\item F.\ Shen, W.\ Tan, Y.\ Zhao, T.\ Masuoka, The Rayleigh--Stokes problem for a heated generalized second grade fluid with fractional derivative model, {Nonlinear Anal.\ RWA} 7 (2006) 1072--1080.
\item C.\ Fetecau, T.\ Hayat, C.\ Fetecau, N. Ali, Unsteady flow of a second grade fluid between two side walls perpendicular to a plate, {Nonlinear Anal.\ RWA} 9 (2008) 1236--1252.
\item T.\ Hayat, C.\ Fetecau, M.\ Sajid, Analytic solution for MHD Transient rotating flow of a second grade fluid in a porous space, {Nonlinear Anal.\ RWA} 9 (2008) 1619--1627.
\item M.\ Khan, S.\ Wang, Flow of a generalized second-grade fluid between two side walls perpendicular to a plate with a fractional derivative model, {Nonlinear Anal.\ RWA} 10 (2008) 203--208.
\item C.\ Xue, J.\ Nie, Exact solution of Stokes' first problem for heated generalized Burgers' fluid in a porous half-space, {Nonlinear Anal.\ RWA} 9 (2008) 1628--1637.
\item M.\ Hussain, T.\ Hayat, S.\ Asghar, C.\ Fetecau, Oscillatory flows of second grade fluid in a porous space, {Nonlinear Anal.\ RWA} 10 (2008) 2403--2414.
\item M.\ Khan, The Rayleigh--Stokes problem for an edge in a viscoelastic fluid with a fractional derivative model, {Nonlinear Anal.\ RWA} 10 (2008) 3190--3195.
\item Y.\ Yao, Y.\ Liu, Some unsteady flows of a second grade fluid over a plane wall, Nonlinear Anal.\ RWA 11 (2010) 4442--4450.
\end{enumerate}
The reader may verify that the above papers are erroneous by specializing the Fourier--Laplace domain representation of the governing equation given therein to the cases for which the correct form is known from \cite{ChristovChristov,ChristovJordan,Jordan}, then showing that the two disagree. Note that, at this time, only the paper in Item 5 above has been formally corrected thanks to the Comment paper of Jordan \cite{Jordan}.

What is more, the recent paper [L.\ Zheng, Y.\ Liu, X.\ Zhang, Slip effects on MHD flow of a generalized Oldroyd-B fluid with fractional derivative, doi:10.1016/j.nonrwa.2011.02.016] purports to study once again a known linear problem. However, the article begins with the incorrect claim that ``$\nabla\bm{V}=0$, $\rho\mathrm{d}\bm{V}/\mathrm{d}t=\nabla\mathbf{T}+\rho\bm{b}$'' (instead of $\nabla\cdot\bm{V}=0$, $\rho\mathrm{d}\bm{V}/\mathrm{d}t=\nabla\cdot\mathbf{T}+\rho\bm{b}$) are the ``constitutive equations'' (instead of the equations of conservation of mass and momentum) of an incompressible fluid. Skipping over the technical problems in between, consider the final conclusion of the paper: ``the results indicated that the strongest shear stress occurs near the plate and the shear stress decreases rapidly with the increase of distance from the plate.'' Of course, this is nothing more than the imposed boundary conditions of a suddenly-moved plate and the decay of the velocity at infinity. It is unclear why the satisfaction of the imposed conditions can be the major conclusion of a scientific paper.

So it goes.

\end{document}